\documentclass[10pt, conference, letterpaper]{IEEEtran}
\IEEEoverridecommandlockouts \makeatletter
\def\ps@headings{%
\def\@oddhead{\mbox{}\scriptsize\rightmark \hfil \thepage}%
\def\@evenhead{\scriptsize\thepage \hfil \leftmark\mbox{}}%
\def\@oddfoot{}%
\def\@evenfoot{}}
\makeatother
\usepackage[dvips]{graphicx}
\usepackage{color}
\usepackage{times}
\usepackage{fancyhdr}
\usepackage{amssymb,amsmath}
\usepackage{lastpage}
\usepackage{epsfig}
\usepackage{algorithm2e}
\usepackage[compress]{cite}

\newtheorem{theorem}{Theorem}[section]

\newtheorem{corollary}[theorem]{Corollary}

\makeatletter
\def\ps@headings{%
\def\@oddhead{\mbox{}\scriptsize\rightmark \hfil \thepage}%
\def\@evenhead{\scriptsize\thepage \hfil \leftmark\mbox{}}%
\def\@oddfoot{}%
\def\@evenfoot{}}
\makeatother 

\begin{document}
\title{Localized Dimension Growth in Random Network Coding: A Convolutional Approach}
\author{\authorblockN{Wangmei Guo, Ning Cai}
\authorblockA{The State Key Lab. of ISN,\\
Xidian University, Xi'an, China\\
Email: \{wangmeiguo, caining\}@mail.xidian.edu.cn} \and
\authorblockN{Xiaomeng Shi, Muriel M\'{e}dard}
\authorblockA{Research Laboratory of Electronics,\\
MIT, Cambridge, USA\\
Email: \{xshi, medard\}@mit.edu}
\thanks{This work is funded by the National Science Foundation (NSF) under grant No.60832001, and partially presented in conference AEW 2010.}
}

\maketitle

\normalsize 

\begin{abstract} We propose an efficient \textit{Adaptive Random Convolutional
Network Coding} (ARCNC) algorithm to address the issue of field size in random network coding. ARCNC operates as a convolutional code,
with the coefficients of local encoding kernels chosen randomly over a small finite field. The lengths of local encoding kernels increase
with time until the global encoding kernel matrices at related sink nodes all have full rank. Instead of estimating the necessary field
size a priori, ARCNC operates in a small finite field. It adapts to unknown network topologies without prior knowledge, by
locally incrementing the dimensionality of the convolutional code. Because convolutional codes of different constraint lengths can coexist in different portions of the network, reductions in decoding delay and memory overheads can be achieved with ARCNC. We show through analysis that this method performs no worse than random linear network codes in general networks, and can provide significant gains in terms of average decoding delay in combination networks.
\end{abstract}



\section{Introduction}\label{sec:introduction}
Since its introduction \cite{ACLY2000}, network coding has been shown to offer advantages in throughput, power consumption, and security in
both wireline and wireless networks. Field size and adaptation to unknown topologies are two of the key issues in network coding. Li et al.
showed constructively that the max-flow bound is achievable by linear algebraic network coding (ANC) if the field is sufficiently large for
a given deterministic multicast network \cite{LYC2003}, while Ho et al. \cite{HMK2006} proposed a distributed random linear algebraic
network code (RLNC) construction that achieves the multicast capacity with probability $(1-d/q)^\eta$, where $\eta$ is the number of
links with random coefficients. Because of the construction simplicity and the ability to adapt to unknown topologies, RLNC are
often preferred over deterministic network codes. While the construction in \cite{HMK2006} allows cycles, which lead to the creation of
convolutional codes, it does not make use of the convolutional nature of the resulting codes to lighten bounds on field size, which may
need to be large to guarantee that decoding succeeds at all sink nodes. Both block network codes (BNC) \cite{medard2003coding,JSCEEJT2005}
and convolutional network codes (CNC) \cite{OCNC2006,EF2004} can mitigate field size requirements. M\'edard et al.
introduced the concept of vector, or block network codes (BNC) \cite{medard2003coding}, and Xiao
et al. proposed a deterministic binary BNC to solve the combination network problem \cite{xiao2008binary}. BNC can operate on
smaller finite fields, but the block length used may need to be pre-determined according to the network size. In
discussing cyclic networks, both Li et al. and Ho et al. pointed out the equivalence between ANC in cyclic networks with delays, and
CNC \cite{LYC2003,HMK2006}. Because of coding introduced across the temporal domain, CNC in general does not have a field size constraint. 
Even though degree growth of the encoding kernel may lead to high computation complexity during decoding when there are many coding nodes along a path to a sink, it is often possible to achieve the network coding advantage by coding at a subset of nodes only \cite{kim2006minimizing}. In addition, the structure of CNC allows decoding to occur symbol-by-symbol, thus offer gains in decoding delay.  However, it may require long coding kernels when the network is unknown. As discussed by Jaggi et al. \cite{jaggi2004linear}, there exists equivalence relationships between ANC, BNC, and CNC. Nonetheless, all three schemes require some prior knowledge on the network topology. Overestimation for the worst case assumption can be wasteful, leading to high computation complexity, decoding delay, and memory overheads.

Our work extends the RLNC and CNC setup, allowing nodes to locally grow the dimensionality of the code until necessary. We propose an
efficient adaptive random convolutional network code (ARCNC) for multicast networks, with local encoding kernels chosen randomly
from a small field, and the code constraint length incremented locally at each node. Our scheme inherits the small field size
property of CNC, while taking advantage of the distributive nature of RLNC. The gains offered by ARCNC are three-fold.
First, it operates in a small finite field. Second, it adapts to unknown network topologies without prior knowledge.
Last, the localized adaptation allows convolutional codes with different code lengths to coexist in different portions of the network, leading to reduction in decoding delay and memory overheads associated with using a pre-determined field size or code length. 

The remainder of this paper is organized as follows: the ARCNC algorithm is proposed in Section~\ref{sec:algorithm} and its performance
analyzed in Section~\ref{sec:analysis}. As an example, the advantages of ARCNC is considered in a combination network in
Section~\ref{sec:examples}. Section~\ref{sec:conclusion} concludes the paper.

\section{Adaptive Randomized Convolutional Network Coding Algorithm}\label{sec:algorithm}
\subsection{Basic Model and Definitions}\label{sec:basicDefs}

We first introduce definitions which will be used throughout the paper. We model a communication network as a finite directed multigraph, denoted by $\mathcal{G}=(\mathcal{V},\mathcal{E})$. An
edge represents a noiseless communication channel on which one symbol is transmitted per unit time. In this paper, we consider the
multicast case. The source node is denoted by $s$, and the set of $d$ sink nodes is denoted by $T=\{r_1,\ldots,r_d\} \subset V$. For every
node $v \in \mathcal{V}$, denote the sets of incoming and outgoing channels to $v$ by $In(v)$ and $Out(v)$. An ordered pair $(e',e)$ of channels is called an adjacent pair when there exists a node $v$ with $e'\in In(v)$ and $e\in Out(v)$.

The symbol alphabet is represented by a base field, $\mathbb{F}_q$, throughout the paper. Assume source $s$ generates  a
message per unit time, consisting of a fixed number $m$ of symbols represented by an $m$-dim row vector $x \in \mathbb{F}_q^m$. We index
time $t$ to start from 0, hence $(t+1)$-th coding around occurs at time $t$. The message transmitted by $s$ can be represented by a power series $x(z)=\sum_{t\geq 0}{x_t z^{t}}$, where $x_t \in F_q^m$ is the message generated at time $t$ and $z$ denotes a unit-time delay. The data propagated over a channel $e$ is $y_e(z)$, a linear function of the source message, and $y_e(z)= x(z) f_e(z)$, where the $m$-dim column vector of rational power series, $f_e(z) = \sum_{t\geq0}f_{e,t}z^{t}$, is called the global encoding kernel over channel $e$. Viewed locally, $y_e(z)$ is a linear combination of messages over all incoming channels, represented as $y_e(z) = \sum_{e' \in In(v)}{k_{e',e}(z)y_{e'}(z)}$, where $k_{e',e}(z)= \sum_{t\geq0}k_{e',e,t}z^{t}$ is the local encoding kernel over the adjacent pair
$(e',e)$. Hence, $f_e(z) = \sum_{e' \in In(v)}{k_{e',e}(z)f_{e'}(z)}$. As discussed in \cite{OCNC2006}, $k_{e',e}(z)$ and $f_e(z)$ are
rational power series in the form of $\frac{p(z)}{1+zq(z)}$, where $p(z)$ and $q(z)$ are polynomials. Collectively,
we call the $|In(v)|\times|Out(v)|$ matrix $K_v(z)=(k_{e',e}(z))_{e'\in In(v), e\in Out(v)}$ the local encoding kernel matrix at node
$v$, and the $m \times |In(v)|$ matrix $F_r(z)=(f_{e}(z))_{e\in In(r)}$ the global encoding kernel matrix at sink $r$.

\subsection{Algorithm Statement for Acyclic Networks}\label{sec:algAcyclic}

\subsubsection{Encoding} at time 0, all local and global encoding kernels are set to 0. Source $s$ generates a
message $x = \sum_{t\geq0}x_t z^{t}$, where $x_t$ consists of $m$ symbols $(x_{t,1},x_{t,2},\cdots,x_{t,m})$. For each intermediate
node $v$, when a symbol is received on $e' \in In(v)$ at time $t$, it stores the symbol in memory as $y_{e',t}$, and chooses the $(t+1)$-th term $k_{e',e,t}$ of the local encoding kernel $k_{e',e}(z)$ uniformly randomly from $\mathbb{F}_q$ for $e \in Out(v)$. Node $v$ assigns registers to store the local encoding kernels, and form the outgoing symbol as

\begin{equation*}
y_{e,t}  = \sum_{e'\in In(v)}\left(\sum_{i=0}^t k_{e',e,i}y_{e',t-
i}\right)\,. \label{eq:yet}
\end{equation*}
In other words, the outgoing symbol is a random linear combination of symbols in the node's memory. The
$(t+1)$-th term of $f_e(z)$, $f_{e,t}$, is placed in the header of the outgoing message.

\subsubsection{Decoding}
at each time instant $t$, each sink node $r$ decides whether its global encoding kernel matrix is full rank. If so, it sends an ACK signal to its parent node. An intermediate node $v$ which has received ACKs from all its children at a time $t_0$ will send an ACK to its parent, and set all subsequent local encoding kernel coefficients $k_{e',e,t}$ to $0$ for all $t>t_0$, $e' \in In(v)$, and $e\in
Out(v)$. In other words, the constraint length of the local convolutional code increases until it is sufficient for downstream sinks to decode. Such automatic adaptation eliminates the need for estimating the field size or the constraint length a priori. It also allows nodes within the network to operate with different constraint lengths as needed.

Once its global encoding kernel matrix $F_r(z)$ is full rank, a sink node $r$ performs sequential decoding as introduced by Erez et al.
\cite{EF2004} to obtain the source message symbol-by-symbol. If $F_r(z)$ is not full rank, $r$ stores received messages and wait for more data to arrive. At time $t$, the algorithm is considered successful if all sink nodes can decode. At sink $r$, the local and global encoding kernels are $k_{e',e}(z)=k_{e',e,0}+k_{e',e,1}z+\cdots+k_{e',e,t}z^{t}$ and  $f_e(z)=f_{e,0}+f_{e,1}z+\cdots+f_{e,t}z^{t}$ respectively, where $k_{d,e,i}$ and $f_{e,i}$ are the encoding coefficients at time $i$. Sink $r$ can decode successfully if there exists at least $m$ linear independent incoming channels, i.e., the determinant of $F_r(z)$ is a non-zero polynomial. At time $t$, $F_r(z)$ can be written as $F_r(z)=F_0+F_1z+\cdots+F_{t}z^{t}$, where $F_i$ is the global encoding kernel matrix at time $i$. Computing the determinant of $F_r(z)$ at every time instant is complex, so we test instead the following conditions, introduced in \cite{CG2009,MS1968} to determine decodability at a sink $r$. The first condition is necessary, while the second is both necessary and sufficient.

\begin{enumerate}
\item $rank\left( {\begin{array}{*{20}c}
       {F_0} & {F_1} &  \cdots  & {F_{t}}  \\
    \end{array}} \right) = m$
\item $rank(M_{t})-rank(M_{t-1})=m \,,$
where
\begin{equation*}
     M_i = \left( {\begin{array}{*{20}c}
     F_0 & F_1    & \cdots & F_{i}  \\
      0  & \ddots & \ddots & \vdots \\
      0  & \cdots &   F_1  & F_1    \\
      0  & \cdots &    0   & F_0    \\
\end{array}} \right).
\end{equation*}
\end{enumerate}
Each sink $r$ checks the two conditions in order. If both pass, $r$ sends an ACK signal to its parent; otherwise, it waits for more data to arrive. Observe that as time progresses, $F_r(z)$ grows in size, until decodability is achieved. This scheme for verifying the invertibility of $F_r(z)$ is based on Massey's theory of decodable convolutional codes, which
transfers the determinant calculation of a polynomial matrix into the rank computation of extended numerical matrices. We do not elaborate on the details and refer interested readers to the original work \cite{MS1968}.


\subsection{Algorithm Statement for Cyclic Networks}

In a cyclic network, a sufficient condition for a convolutional code to be successful is that the constant coefficient matrix consisting of all local encoding kernels be nilpotent \cite{KM2003,CG2009}; it can be shown that this condition is satisfied if we code over an acyclic topology at time 0 \cite{CG2009}. Instead of flooding the network starting from the source, at time $0$, we first find $m$ disjoint paths for each sink node from the source. For adjacent pairs along any of these paths, local encoding kernels $k_{e',e,0}$ are chosen uniformly randomly over $\mathbb{F}_q$; for edges not along such paths, local encoding kernels are set to zero at time 0. For intermediate nodes, outgoing message symbols are computed as a random linear combination of its incoming received symbols:
$y_{e,0} = \sum_{e' \in In(v)}
{k_{e',e,0} } y_{e',0}  = x_0 f_{e,0}$. After initialization, the algorithm proceeds exactly the same as in the acyclic case.

\section{Analysis of ARCNC}\label{sec:analysis}
\subsection{Success probability}

Discussions in \cite{HMK2006, LYC2003, KM2003} state that in a network with delays, ANC gives rise to random processes which can be written algebraically in terms of a delay variable $z$. In other words, a convolutional code can naturally evolve from the message propagation and the linear encoding process. ANC in the delay-free case is therefore equivalent to CNC with constraint length 1. Similarly, using a CNC with constraint length $l>1$ on a delay-free network is equivalent to performing ANC on the same network, but with $l-1$ self-loops attached to each encoding node. Each self-loop carries $z, z^2, \ldots, z^{l-1}$ units of delay respectively. The ARCNC algorithm we have proposed therefore falls into the framework given by Ho et al. \cite{HMK2006}, in the sense that the convolution process either arises naturally from cycles with delays, or can be considered as computed over self-loops appended to acyclic networks. From \cite{HMK2006}, we have the following theorem,

\begin{theorem}\label{thm:prob}
For multicast over a general network with $d$ sinks, the ARCNC algorithm over $\mathbb{F}_q$
can achieve a success probability of at least $(1-d/q^{t+1})^{\eta}$ at time $t$, if $q^{t+1}>d$, and $\eta$ is the number of links with
random coefficients.
\end{theorem}

\begin{proof} At node $v$, the local encoding kernel $k_{e',e}(z)$ at time $t$ is a polynomial with maximal degree $t$, i.e.,  $k_{e',e}(z)=k_{e',e,0}+k_{e',e,1}z+\cdots+k_{e',e,t}z^{t}$, where $k_{e',e,i}$ is randomly chosen over $\mathbb{F}_q$. If we group the encoding coefficients, the ensuing vector, $k_{e',e}=\{k_{e',e,0},k_{e',e,1},\cdots,k_{e',e,t}\}$, is of
length $t+1$, and corresponds to a random element over the extension field
 $\mathbb{F}_{q^{t+1}}$. Using the result in \cite{HMK2006}, we conclude that the success probability of ARCNC at time $t$ is at least $(1-d/q^{t+1})^{\eta}$, as long as $q^{t+1}>d$.
\end{proof}

We could similarly consider the analysis done by Balli et al. \cite{BYZ2009}, which states that the success probability is at least
$(1-d/(q-1))^{|J|+1}$, $|J|$ being the number of encoding nodes, to show that a tighter lower bound can be given on the success probability
of ARCNC, when $q^{t+1}>d$.

\subsection{Stopping time}\label{sec:stoppingTime}
We define the stopping time $T_i$ for sink $i$, $1\leq i\leq d$, as the time it takes $i$ to achieve decodability. Also denote by $T_N$ the time it takes for all sinks in the network to successfully decode, i.e., $T_N=\max\{T_1,\ldots,T_d\}$. Then we have the following corollary:

\begin{corollary}
For any given $0<\varepsilon<1$, there exists a $T_0>0$ such that for any $t \geq T_0$, ARCNC solves the multicast problem with probability
at least $1-\varepsilon$, i.e., $P(T_N > t)  < \varepsilon$.\label{crlry2}
\end{corollary}
\proof Let $T_0  = \left\lceil \lg_q d-\lg_q(1-\sqrt[\eta]{1-\epsilon}) \right\rceil -1 $, then $T_0 + 1 \geq \lceil \log_q d \rceil$ since
$0 < \varepsilon <1$, and $(1-d/q^{T_0+1})^\eta>1-\varepsilon$. Applying Theorem~\ref{thm:prob} gives $P(T_N > t) \leq P(T_N >T_0) <
1-(1-d/q^{t+1})^{\eta}< \varepsilon$ for any $t\geq T_0$,
\endproof
Corollary~\ref{crlry2} shows that $T_N$ is a valid random variable, and ARCNC converges in a finite amount of time for a multicast connection.

Another relevant measure of the performance of ARCNC is the average stopping time $E[T]=\frac{1}{d}\sum_{i=1}^{d}T_i$. Observe that $E[T] \leq E[T_N]$, where
\begin{align*}
  E[T_N] & =\sum_{i=0}^{\infty}tP(T_N=t)\\
   & = \sum_{t=1}^{\lceil lg_q d\rceil -1}P(T_N\geq t) + \sum_{t= \lceil lg_q d\rceil}^{\infty}{P(T_N\geq t)}\\
         & \le \lceil lg_q d\rceil -1 + \sum_{t=\lceil lg_q d\rceil}^{\infty}[1-(1-\frac{d}{q^t})^{\eta}] \\
         & = \lceil lg_q d\rceil -1 + \sum_{k=1}^{\eta}(-1)^{k-1}\left(\begin{array}{c}
                                                           \eta \\
                                                           k
                                                         \end{array}
     \right) \frac{d^k}{q^{\lceil lg_q d\rceil k}-1}\,.
\end{align*}
When $q$ is large, the summation term becomes $1-(1-d/q)^\eta$ by the binomial expansion.  Hence as $q$ increases, the second term above
diminishes to $0$, while the first term $\lceil lg_q d\rceil -1$ is 0. $E[T]$ is therefore upper-bounded by a term converging to 0; it is
also lower bounded by 0 because at least one round of random coding is required. Therefore, $E[T]$ converges to $0$ as $q$
increases.  In other words, if the field size is large enough, ARCNC reduces in effect to RLNC.

Intuitively, the average stopping time of ARCNC depends on the network topology. In RLNC, field size is determined by the
worst case node. This process corresponds to having all nodes stop at $T_N$ in ARCNC. ARCNC enables each node to decide locally what is a good constraint length to use, depending on side information from downstream nodes. The corresponding effective field size is therefore expected to be smaller than in RLNC. Two possible consequences of a smaller effective field size are reduced decoding delay, and reduced memory requirements.

\subsection{Complexity}
To study the computation complexity of ARCNC, first observe that once the adaptation process terminates, the amount of computation needed
for the ensuing code is no more than a regular CNC. In fact, the expected computation complexity is proportional to the average code
length of ARCNC. We therefore omit the details of the complexity analysis of regular CNC here and refer interested readers to \cite{EF2004}.

For the adaptation process, the encoding operations are described by $f_{e,t}=\sum_{e'\in In(v)}(\sum_{i=0}^{t}k_{e',e,i}f_{e',t-i})$. If the algorithm stops at time $T_N$, then the number of operations in the encoding steps is $O(D_{in}|\mathcal{E}|T_N^2m)$, where $D_{in}$ represents the maximum input degree over all nodes.

To determine decodability at a sink $r$, we check if the rank of the global encoding matrix $F_r(z)$ is $m$. A straight-forward approach is to check whether the determinant of $F_r(z)$ is a non-zero polynomial. Alternatively, Gaussian elimination could also be applied. At time $t$, because $F_r(z)$ is an $m \times |In(r)|$ matrix and each entry is a polynomial with degree $t$, the complexity of checking if $F_r(z)$ is full rank is $O(D_{in}^22^mmt^2)$. Instead of computing the determinant or using Gaussian elimination directly, we propose to check the conditions given in Section~\ref{sec:algAcyclic}. For each sink $r$, at time $t$, determining $rank\left({\begin{array}{*{20}c} F_0 & F_1 & \cdots F_{t} \end{array}}\right)$ requires a computation complexity of $O(D_{in}^2mt^2)$. If the first test passes, we then need to calculate $rank(M_{t})$ and $rank(M_{t-1})$. Observe that $rank(M_{t-1})$ was computed during the last iteration. $M_t$ is a $(t+1)|In(r)| \times (t+1)|In(r)|$ matrix over field $\mathbb{F}_q$. The complexity of calculating $rank(M_t)$ by Gaussian elimination is $O(D_{in}^2mt^3)$. The process of checking decodability is performed during the adaptation process only, hence the computation complexity here can be amortized over time after the coding coefficients are determined. In addition, as decoding occurs symbol-by-symbol, the adaptation process itself does not impose any additional delays.

\section{Examples}\label{sec:examples}
ARCNC adapts to the topology of general networks by locally increasing the convolutional code length, and generating coding coefficients
randomly. Such adaptation allows nodes to code with different lengths, thus possibly reducing decoding delay and memory overheads associated with overestimating according to the worst case. As examples, next we consider a small combination network to illustrate how ARCNC operates, and how delay and memory overheads can be measured. We also consider a more general combination network to show that ARCNC can obtain significant gains in terms of decoding delay in this case.

A $n\choose m$ combination network contains a single source $s$ that multicasts $m$ independent messages over $\mathbb{F}_q$ through $n$ intermediate nodes to $d$ sinks \cite{NWT2005}; each sink is connected to a distinct set of $m$ intermediate nodes, $d={n\choose m}$. Assuming unit capacity links, the min-cut to each sink is $m$. It can be shown that in combination networks, routing is insufficient and network coding is needed to achieve the multicast capacity $m$.

\subsection{A $4 \choose 2$ combination network}

Fig.~\ref{fig:42ex} illustrates a simple $4\choose 2$ combination network. To see how ARCNC operates, let the messages generated by source $s$ be $\sum_{t=0}^{\infty}(a_t, b_t)z^t$. Assume field size is $q=2$. Observe that only $s$ is required to code; intermediate nodes relay on received messages directly. At time 0, $s$ chooses randomly the local encoding kernel matrix. Suppose the
realization is
\[K_s(z)|_{t=0}=\left(
               \begin{array}{cccc}
                 1 & 0 & 1 & 1 \\
                 0 & 1 & 1 & 1 \\
               \end{array}
             \right)\,.
\]
The first 5 sinks can therefore decode directly at time 0, but sink $r_6$ cannot. Therefore, at time $1$, $s$ increases the
convolutional code length for the right two intermediate nodes. Suppose the updated local encoding kernel is
\[K_s(z)|_{t=1}=\left(
               \begin{array}{cccc}
                 1 & 0 & 1 & 1+z \\
                 0 & 1 & 1 & 1 \\
               \end{array}
             \right)\,.
\]
Sink $r_6$ is now capable of decoding. It therefore acknowledges its parents, which in turn causes $s$ to stop incrementing the corresponding code length. By decoding sequentially, $r_6$ can recover messages $(a_0,b_0)$ at time 1, $(a_1,b_1)$ at time 2, and $(a_{t-1},b_{t-1})$ at time
$t$. Observe that for sinks $r_2$ to $r_5$, which are also connected to the two right intermediate nodes, the global encoding kernels
increases in length until time 1 as well. In other words, these sinks decode with minimal delay of 0, but require twice the memory when
compared to $r_1$.

\begin{figure}[t!]
  \centering
  \includegraphics[width=8.5cm]{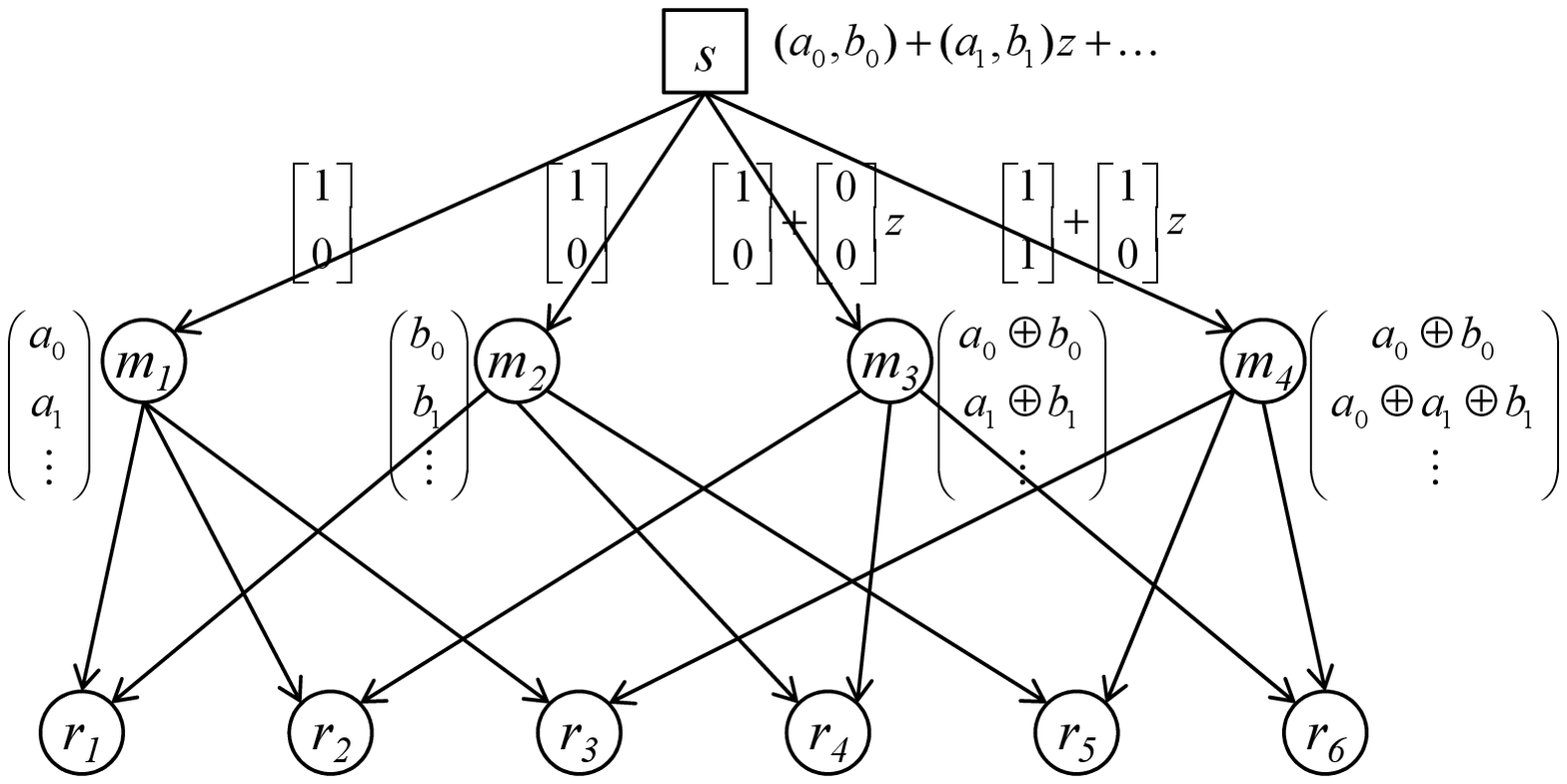}\\
  \caption{A $4\choose 2$ combination network.}
  \label{fig:42ex}
\end{figure}

In this example, at the single encoding node $s$, the code lengths used are $(1,1,2,2)$, with an average of $3/2$. At the sinks, the decoding delays are $(0,0,0,0,0,1)$, with an average of $1/6$. For the same ${n\choose 2}$ combination network, the deterministic BNC algorithm given by Xiao et al. \cite{xiao2008binary}, designed specifically for combination networks, requires an average decoding delay of 1, since the entire block of data needs to be received before decoding, and the block length is 2. In terms of memory, the BNC algorithm requires 4 bits per node to store data, while ARCNC requires $2$ bits at $r_1$, and $4$ bits at all other nodes, with an overall average of $\frac{42}{11}$. In other words, ARCNC learns the topology of this network automatically, and achieves much lower decoding delays without compromising the amount of memory needed, when compared to the deterministic BNC algorithm designed specifically for combination networks.

Furthermore, if RLNC is used, even for a single sink $r_1$ to achieve decodability a at time $0$, the field size needs to be a minimum of $2^3$ for the decoding probability to be $\frac{49}{64}$. More coding rounds will be needed to achieve a higher success probability. In other words, with RLNC over $\mathbb{F}_{2^3}$, the average decoding delay will be higher than 1, and the amount of memory needed is at the minimum 6 bits per node. In other words, ARCNC adapts to the network topology at significantly lower costs than RLNC, achieving both lower decoding delays and lower memory overheads, while operating with lower complexities in a smaller field.

\subsection{Decoding delays in a $n \choose m$ combination network}
We now consider a general $n \choose m$ combination network, and show that the average decoding delay can be
significantly improved by ARCNC  when compared to the deterministic BNC algorithm. Recall from definitions in Section~\ref{sec:stoppingTime} that the average decoding delay is the average stopping time $E[T] = E\left[ \frac{1}{d}\sum_{i=1}^{d} T_i \right] =  E[T_i]$. At time $t-1$, for sink node $i$, if it has not stopped increasing the constraint length of the convolution code, the global encoding kernel is a $m\times m$ matrix of degree $t-1$ polynomials in $z$ over $\mathbb{F}_q$. This matrix has full rank with probability $Q=(q^{tm}-1)(q^{tm}-q^t)\cdots(q^{tm}-q^{t(m-1)})/q^{tm^2}$, so the probability that sink $i$ decodes after time $t-1$ is $ P(T_i\geq t) = 1- Q$. The average stopping time over all sink nodes is then upper bounded by
\begin{align}
E[T] & = E[T_i]=\sum_{t=1}^\infty P(T_i\geq t) < \sum_{t=1}^\infty (1-(1-\frac{1}{q^{t}})^m) \notag\\
     & = \sum_{k=1}^m (-1)^{k-1} \left(\begin{array}{c} m
\\ k \\ \end{array} \right) \frac{1}{q^k-1} \triangleq ET_{UB} \label{eq:ET}
\end{align}

First observe that $E[T]$ a function of $m$ and $q$, independent of the value of $n$.
In other words, if $m$ is fixed, but $n$ increases, the expected decoding delay does not change. Next observe that if $q$ is large,
$ET_{UB}$ becomes 0, consistent with the general analysis in Section~\ref{sec:stoppingTime}.

A similar upper bound can be found for the variance of $T$ as well. It can be shown that
\begin{align}
var(T)  & = \frac{E[T_i^2]}{d} + \frac{1}{d^2}\left[\sum_{i=1}^{d}\sum_{j\neq i}E[T_iT_j]\right] - E^2[T_i]\notag\\
& < \frac{ET^2_{UB}}{d} + \frac{m}{n}\rho_{UB} - (1+\frac{m}{n})(ET_{LB})^2\,, \label{eq:varT}
\end{align}

\noindent where $ET^2_{UB}$ is an upperbound for $E[T_i^2]$, $\rho_{UB}$ is an upperbound for $E[T_iT_j]_{i\neq j}$, and $ET_{LB}$ is a lowerbound for $E[T_i]$. Details of the derivation are given in the Appendix. All three quantities are functions of $m$ and $q$, independent of $n$. If $m$ and $q$ are fixed, as $n$ increases, $d$ also increases, and $var(T)$ diminishes to 0. Combining this result with a bounded expectation, what we can conclude is that even if more intermediate nodes are added, a large proportion of the sink nodes can still be decoded within a small number of coding rounds. On the other hand, if $m$ and $n$ are comparable in scale, for example, if $m=n/2$, then the bound above depends on the exact value of $ET^2_{UB}$, $\rho_{UB}$ and $ET_{UB}$. We leave the detailed analysis of this case for future studies.

Comparing with the deterministic BNC algorithm proposed by Xiao et al. \cite{xiao2008binary}, we can see that for a large combination network, with fixed $q$ and $m$, ARCNC achieves much lower decoding delay. In the BNC scheme, the block length is required to be
$p\geq n-m$ at the minimum. Thus the decoding delay increases at least linearly with $n$. Similar comparisons can be made with RLNC, and it is not hard to see that we can obtain gains in both decoding delay and memory.

So far we have used $n\choose m$ combination networks explicitly as an example to illustrate the operations and the decoding delay gains of ARCNC. It is important to note, however, that this is a very special network, in which only the source node is required to code, and each sink  shares at least $1$ parent with other ${n-1\choose m-1}$ sinks. If sink $r$ cannot decode, all other ${n-1\choose m-1}$ sinks sharing parents with $r$ are required to increase their memory capacity. Therefore, in combination networks, we do not see much gains in terms of memory overheads when compared with BNC algorithms. In more general networks, however, when sinks do not share ancestors with as many other sinks, ARCNC can achieve gains in terms memory overheads as well, in addition to decoding delay. Due to space limitations, we do not give any detailed analysis, but it can be shown, for example, that in an umbrella-shaped network, memory overheads can be significantly reduced with ARCNC when compared to other network codes.

\section{Conclusion}\label{sec:conclusion}
We propose an adaptive random convolutional network code (ARCNC), which operates in a small field, and locally and automatically adapts to the network topology by incrementally growing the constraint length of the convolution process. We show through analysis that ARCNC performs no worse than random algebraic linear network codes, and illustrate through a combination network example that it can reduce the average decoding delay significantly. ARCNC can also reduce memory overheads in networks where sinks do not share the majority of their ancestors with other sinks. One possible future direction of analysis is to characterize the behavior of this algorithm over random networks, the results of which will enable us to decide on the applicability of ARCNC to practical systems.

\bibliographystyle{IEEEtran}
\bibliography{References}

\begin{thebibliography}{10}
\providecommand{\url}[1]{#1}
\csname url@samestyle\endcsname
\providecommand{\newblock}{\relax}
\providecommand{\bibinfo}[2]{#2}
\providecommand{\BIBentrySTDinterwordspacing}{\spaceskip=0pt\relax}
\providecommand{\BIBentryALTinterwordstretchfactor}{4}
\providecommand{\BIBentryALTinterwordspacing}{\spaceskip=\fontdimen2\font plus
\BIBentryALTinterwordstretchfactor\fontdimen3\font minus
  \fontdimen4\font\relax}
\providecommand{\BIBforeignlanguage}[2]{{%
\expandafter\ifx\csname l@#1\endcsname\relax
\typeout{** WARNING: IEEEtran.bst: No hyphenation pattern has been}%
\typeout{** loaded for the language `#1'. Using the pattern for}%
\typeout{** the default language instead.}%
\else
\language=\csname l@#1\endcsname
\fi
#2}}
\providecommand{\BIBdecl}{\relax}
\BIBdecl

\bibitem{ACLY2000}
R.~Ahlswede, N.~Cai, S.~Li, and R.~Yeung, ``{Network information flow},''
  \emph{IEEE Trans. on Inform. Theory}, vol.~46, no.~4, pp. 1204--1216, 2002.

\bibitem{LYC2003}
S.~Li, R.~Yeung, and N.~Cai, ``{Linear network coding},'' \emph{IEEE Trans. on
  Inform. Theory}, vol.~49, no.~2, pp. 371--381, 2003.

\bibitem{HMK2006}
T.~Ho, M.~M{\'e}dard, R.~Koetter, D.~Karger, M.~Effros, J.~Shi, and B.~Leong,
  ``{A random linear network coding approach to multicast},'' \emph{IEEE Trans.
  on Inform. Theory}, vol.~52, no.~10, pp. 4413--4430, 2006.

\bibitem{medard2003coding}
M.~M{\'e}dard, M.~Effros, D.~Karger, and T.~Ho, ``{On coding for non-multicast
  networks},'' in \emph{Proc. of the 41st Allerton Conference}, vol.~41, no.~1,
  2003, pp. 21--29.

\bibitem{JSCEEJT2005}
S.~Jaggi, P.~Sanders, P.~Chou, M.~Effros, S.~Egner, K.~Jain, and L.~Tolhuizen,
  ``{Polynomial time algorithms for multicast network code construction},''
  \emph{IEEE Trans. on Inform. Theory}, vol.~51, no.~6, pp. 1973--1982, 2005.

\bibitem{OCNC2006}
S.~Li and R.~Yeung, ``{On convolutional network coding},'' in \emph{Proc. of
  IEEE Int. Sym. on Info. Theory}, 2006, pp. 1743--1747.

\bibitem{EF2004}
E.~Erez and M.~Feder, ``{Convolutional network codes},'' in \emph{Proc. of IEEE
  Int. Sym. on Info. Theory (ISIT)}, 2005, p. 146.

\bibitem{xiao2008binary}
M.~Xiao, M.~M{\'e}dard, and T.~Aulin, ``{A binary coding approach for
  combination networks and general erasure networks},'' in \emph{Proc. of IEEE
  Int. Sym. on Info. Theory (ISIT)}, 2008, pp. 786--790.

\bibitem{kim2006minimizing}
M.~Kim, C.~Ahn, M.~M{\'e}dard, and M.~Effros, ``{On minimizing network coding
  resources: An evolutionary approach},'' in \emph{Proc. NetCod}, 2006.

\bibitem{jaggi2004linear}
S.~Jaggi, M.~Effros, T.~Ho, and M.~Medard, ``{On linear network coding},'' in
  \emph{Proc. of the 42nd Allerton Conference}, 2004.

\bibitem{CG2009}
N.~Cai and W.~Guo, ``{The conditions to determine convolutional network coding
  on matrix representation},'' in \emph{Proc. NetCod}, 2009, pp. 24--29.

\bibitem{MS1968}
J.~MasseyY and M.~Sain, ``{Inverses of linear sequential circuits},''
  \emph{IEEE Trans. on Comp.}, vol. 100, no.~4, pp. 330--337, 2006.

\bibitem{KM2003}
R.~Koetter and M.~M{\'e}dard, ``{An algebraic approach to network coding},''
  \emph{IEEE/ACM Trans. on Networking}, vol.~11, no.~5, pp. 782--795, 2003.

\bibitem{BYZ2009}
H.~Balli, X.~Yan, and Z.~Zhang, ``{On randomized linear network codes and their
  error correction capabilities},'' \emph{IEEE Trans. on Inform. Theory},
  vol.~55, no.~7, pp. 3148--3160, 2009.

\bibitem{NWT2005}
R.~Yeung, S.~Li, N.~Cai, and Z.~Zhang, ``{Network Coding Theory: Single
  Sources},'' \emph{Foundations and Trends{\textregistered} in Communications
  and Information Theory}, vol.~2, no.~4, pp. 241--329, 2005.

\end{thebibliography}
\section{Appendix}\label{sec:appendix}
In what follows, we derive the inequalities shown in Eq.~(\ref{eq:ET}) and (\ref{eq:varT}), which are used in analyzing the decoding delay gains of ARCNC in combination networks.

Consider a general $n \choose m$ combination network. At time $t-1$, for sink node
$i$, if it has not stopped increasing the constraint length of the convolution code, the global encoding kernel is a $m\times m$ matrix of
degree $t-1$ polynomials in $z$ over $\mathbb{F}_q$. This matrix has full rank with probability
$Q=(q^{tm}-1)(q^{tm}-q^t)\cdots(q^{tm}-q^{t(m-1)})/q^{tm^2}$, so the probability that sink $i$ decodes after time $t-1$ is
\begin{align*}
P(T_i\geq t) & = 1- Q = 1- \prod_{l=1}^m(1-\frac{1}{q^{tl}}) \, \quad t\geq 0. \label{eq:pt}\\
\end{align*}
The average stopping time for sink node $i$ is therefore bounded as follows
\begin{align*}
E[T_i] & = \sum_{t=1}^\infty tP(T_i = t) = \sum_{t=1}^\infty P(T_i\geq t)\\
     & = \sum_{t=1}^\infty (1-\prod_{i=1}^m(1-\frac{1}{q^{ti}}))  \\
     & < \sum_{t=1}^\infty (1-(1-\frac{1}{q^{t}})^m)\\
     & = \sum_{t=1}^\infty (1-\sum_{k=0}^m(-1)^k\left(\begin{array}{c} m \\ k \\ \end{array} \right)(\frac{1}{q^t})^k)\\
     & = \sum_{k=1}^m(-1)^{k-1}\left(\begin{array}{c} m \\ k \\ \end{array} \right)(\sum_{t=1}^{\infty}\frac{1}{q^{tk}})\\
     & = \sum_{k=1}^m (-1)^{k-1} \left(\begin{array}{c} m \\ k \\ \end{array} \right) \frac{1}{q^k-1} \\
     & \triangleq ET_{UB}
\end{align*}
\begin{align*}
E[T_i] & =  \sum_{t=1}^\infty(1-\prod_{l=1}^m(1-\frac{1}{q^{tl}})) \\
     &  > \sum_{t=1}^\infty (1-(1-\frac{1}{q^{tm}})^m)\\
     & = \sum_{k=1}^m (-1)^{k-1} \left(\begin{array}{c} m \\ k \\ \end{array} \right) \frac{1}{q^{km}-1} \\
     & \triangleq ET_{LB}
\end{align*}
The average stopping time over all sink nodes is thus
\begin{align*}
E[T] & = E\left[ \frac{1}{d}\sum_{i=1}^{d} T_i \right] = E[T_i]\\
& < \sum_{k=1}^m (-1)^{k-1} \left(\begin{array}{c} m
\\ k \\ \end{array} \right) \frac{1}{q^k-1} = ET_{UB}
\end{align*}

To bound the variance of $T$, first observe that for a given sink $r$, the number of other sinks which share at least one parent with $r$ is
$$\Delta=\left(\begin{array}{c} n-1 \\ m-1 \\ \end{array} \right)\,.$$ 
Since $d =\left(\begin{array}{c} n \\ m \\ \end{array} \right)$, we have
\begin{align*}
\frac{\Delta}{d}=\frac{{(n - 1)!}}{{(m - 1)!\left( {n - m} \right)!}} \cdot \frac{{m!\left( {n - m} \right)!}}{{n!}}=\frac{m}{n}\,,
\end{align*}
and
\begin{align*}
var(T)   & = E[T^2] - E^2[T] \\
         & = E\left[\left(\frac{1}{d}\sum_{i=1}^d T_i\right)^2\right] - E^2[T_i]\\
         & = \frac{E[T_i^2]}{d} + \frac{1}{d^2}\left[\sum_{i=1}^{d}\sum_{i\neq j}E(T_iT_j)\right] - E^2[T_i]\\
\end{align*}
To calculate $E[T_i^2]$ and $E[T_iT_j]$, first let $Y=T_i^2$. Since $\sum_{t=1}^\infty tp^t=\frac{1}{(1-p)^2}$ for any $0<p<1$, we have

\begin{align*}
E[T_i^2] & = \sum_{t=1}^\infty t^2P(T_i=t) \\
         & = \sum_{y=1}^\infty yP(Y=y) = \sum_{y=1}^\infty P(Y\geq y)\\
         & = \sum_{t=1}^\infty ((t+1)^2-t^2)P(T_i\geq t) \\
         & = \sum_{t=1}^\infty (2t+1) P(T_i\geq t) \\
         & < \sum_{t=1}^\infty (2t+1) (1-(1-\frac{1}{q^{t}})^m) \\
         & < ET_{UB} + 2 \sum_{k=1}^m (-1)^{k-1} \left(\begin{array}{c} m \\ k \\ \end{array} \right) \sum_{t=1}^\infty \frac{t}{q^{tk}}\\
         & = ET_{UB}  + 2 \sum_{k=1}^m (-1)^{k-1} \left(\begin{array}{c} m \\ k \\ \end{array} \right) (\frac{q^k}{q^k-1})^2 \\
         & \triangleq ET^2_{UB}
\end{align*}
Let $\rho_\lambda=E[T_iT_j]$ if sink nodes $i$ and $j$ share $\lambda$ intermediate nodes as parents, $0 \leq \lambda < m$. Then, $\rho_0 = E^2[T_i]$. When $\lambda \neq 0$, given $i$ succeeds at time $t_1$, the probability that sink $j$ has full rank before $t_2$ is
\begin{align*}
P(T_j < t_2|T_i=t_1) > \prod_{l=1}^{m-\lambda}(1-\frac{1}{q^{t_2l}}) >(1-\frac{1}{q^{t_2}})^{m-\lambda}\,.
\end{align*}
Hence,
\begin{align*}
 \rho_\lambda & = E(T_iT_j) \\
            & = \sum_{t_1=1}^\infty \sum_{t_2=t_1}^\infty t_1 t_2 P(T_i = t_1)P(T_j=t_2|T_i=t_1)\\
            & = \sum_{t_1=1}^\infty t_1 P(T_i = t_1) \sum_{t_2=1}^\infty t_2 P(T_j=t_2|T_i=t_1)\\
            & = \sum_{t_1=1}^\infty t_1 P(T_i = t_1) \sum_{t_2=1}^\infty P(T_j\geq t_2|T_i=t_1) \\
            & < \sum_{t_1=1}^\infty t_1 P(T_i = t_1) \sum_{t_2=1}^\infty (1-(1-\frac{1}{q^{t_2}})^{m-\lambda})\\
            & < \sum_{t_1=1}^\infty t_1 P(T_i = t_1) \sum_{k=1}^{m-\lambda}(-1)^{k-1}\left(\begin{array}{c} m-\lambda \\ k \\ \end{array} \right) \frac{1}{q^k-1}\\
            & < E[T_i](\sum_{k=1}^{m-\lambda}(-1)^{k-1}\left(\begin{array}{c} m-\lambda \\ k \\ \end{array} \right) \frac{1}{q^k-1})\\
            & \triangleq \rho_{\lambda,UB}
\end{align*}
Let $\rho_{UB} = \max\{\rho_{1,UB},\ldots,\rho_{m-1,UB}\}$, then the variance of the average decoding time of ARCNC is bounded as 
\begin{align*}
Var(T)  & < \frac{E[T_i^2]}{d} + \frac{\Delta}{d}\rho_{UB}-
(1+\frac{\Delta}{d})E^2[T_i]\\
& <\frac{ET^2_{UB}}{d} + \frac{\Delta}{d}\rho_{UB}-
(1+\frac{\Delta}{d})(ET_{LB})^2\\
& = \frac{ET^2_{UB}}{d} + \frac{m}{n}\rho_{UB}-
(1+\frac{m}{n})(ET_{LB})^2
\end{align*}
Since all of the upper-bound and lower-bound constants in the above expression are functions of $m$ and $q$ only, if $m$ is fixed and $n$ increases, $d$ also increases correspondingly, and $var(T) \rightarrow 0$.


\end{document}